\documentclass[11pt, a4paper, twocolumn]{IEEEtran}

\usepackage{amsmath,epsfig,amssymb,verbatim,amsopn,cite,subfigure,multirow}
\usepackage{balance}
\usepackage[usenames,dvipsnames]{color}
\usepackage[all]{xy}  
\usepackage{url}
\usepackage{amsfonts}
\usepackage{amssymb}
\usepackage{epsfig}
\usepackage{epstopdf}
\usepackage{slashbox}
\usepackage{bm}
\usepackage[margin=14.0mm,top=19mm,bottom=18mm]{geometry}

\usepackage{multibib}
\usepackage[nodisplayskipstretch]{setspace}
\newtheorem{Lemma}{Lemma}

  {\proof}{\proofend}
\newtheorem{proposition}{Proposition}

\def\ie{{i.e}.}

\DeclareMathOperator*{\argmax}{arg\,max}

\newcommand{\ve}[1]{\boldsymbol{#1}}

 \newcommand{\vc}{\ve{c}}

\newcommand{\Rd}{\mathbb{R}^2}

\newcommand{\drd}{d_{\mathbb{R}D}}

\newcommand{\ro}{x_D}
\newcommand{\Phro}{\Phi^{r}_o}

\newcommand{\lcb}{\eta}

\newcommand{\Prob}{\textnormal{Pr}}

\newcommand{\AuthorOne}{Mohommadali Mohammadi$^\dag$}
\newcommand{\AuthorTwo}{Himal A. Suraweera$^\ddagger$}
\newcommand{\AuthorThree}{Xiangyun Zhou$^\dag$}

\definecolor{light-gray}{gray}{0.65}

\newcounter{mytempeqcounter}

\newcommand{\bigformulatop}[2]{%
  \begin{figure*}[!t]
    \normalsize
    \setcounter{mytempeqcounter}{\value{equation}}
    \setcounter{equation}{#1}
    #2

    \setcounter{equation}{\value{mytempeqcounter}}
    \hrulefill
    \vspace*{4pt}
  \end{figure*}
}

\newcommand{\ThankFour}{This work was done while the first author was visiting the Australian National University.}

\begin{document}

\title{Outage Probability of Wireless Ad Hoc Networks with Cooperative Relaying}

\author{\authorblockN{\AuthorOne,\:\AuthorTwo,\:and\:\AuthorThree\\}
$^\dag$ Research School of Engineering, The Australian National University, Australia \\
$^\ddagger$ Singapore University of Technology and Design, Singapore\\
\small{Email: mohammadali.mohammadi@anu.edu.au, himalsuraweera@sutd.edu.sg, xiangyun.zhou@anu.edu.au  }}\normalsize
\vspace{-5em}
\maketitle
\vspace{-5em}
\begin{abstract}
In this paper, we analyze the performance of cooperative transmissions in wireless ad hoc networks with random node locations. According to a contention probability for message transmission, each source node can either transmits its own message signal or acts as a potential relay for others. Hence, each destination node can potentially receive two copies of the message signal, one from the direct link and the other from the relay link. Taking the random node locations and interference into account, we derive closed-form expressions for the outage probability with different combining schemes at the destination nodes. In particular, the outage performance of optimal combining, maximum ratio combining, and selection combining strategies are studied and quantified.
\let\thefootnote\relax\footnotetext{\ThankFour}
\end{abstract}
\vspace{-2.0em}
\section{Introduction}
Large-scale decentralized wireless systems such as ad hoc networks have attracted significant
attention over the last decade due to their wide range of practical applications. The multiuser
nature of these networks motivates the use of cooperative transmission in which additional links
via relay nodes are established to enhance the quality of communication between the source node and
destination node~\cite{Laneman:IT:Dec:2004}. However, there are two key features, namely
indiscriminate node placement and network interference, which make the design and analysis of
cooperative communication a challenging task in wireless ad hoc networks.

Studies on relay selection in interference-free and deterministic networks have shown that
opportunistic relaying where a single relay is chosen to aid the source-destination transmission
can guarantee a significant performance improvement whilst having low implementation
complexity~\cite{Bletsas:Jsak:2006}. In wireless ad hoc networks, however, interference and random
node locations (e.g., due to high mobility) need to be taken into account in any meaningful
analysis. Relay selection methods based on stochastic geometry models, where the relay locations
follows a Poisson point process (PPP), have been investigated in a few recent
works~\cite{Andrew:VTC:2009,Poor:JWCOM:2011,Ganti:ISIT:2009,QOS:RelaySelection:TVT2011}.
Specifically,~\cite{Andrew:VTC:2009} studied the throughput scaling laws when opportunistic relay
selection is performed and~\cite{Poor:JWCOM:2011} investigated the outage performance of
opportunistic relay selection for an interference-free random network. The authors
in~\cite{Ganti:ISIT:2009} proposed four decentralized relay selection methods based on the
available location information or received signal strength and authors
in~\cite{QOS:RelaySelection:TVT2011} defined a quality of service (QoS) region for relay selection
to guarantee a target QoS at the destination. Common to all these studies is the simplified
assumption that the direct link is neglected or at most selection combining (SC) is used, where the
destination only selects one link from the relay and direct links for data detection. To the best
of our knowledge, the ultimate benefit of cooperative transmission utilizing both the relay and
direct links in such networks has not been studied in the existing literature. Our main goal is to
fill this important gap and enhance the fundamental understanding of cooperative relaying in
large-scale ad hoc networks.

In this paper, we analyze the outage performance of an opportunistic cooperative ad hoc network.
The locations of all potential transmitting nodes are modeled as a PPP. Each potential transmitter
is allowed to transmit its message at a given time slot according to a contention probability.
Hence, the potential transmitters at a given time slot are divided into a group of active source
nodes and a group of idle nodes, where the latter becomes the potential relays for the former. The
decode-and-forward (DF) protocol is assumed at the relays. We consider two relay selection schemes,
namely, best relay selection and random relay selection. For each relay selection scheme, we study
the outage performance of different signal combining methods at the destination, namely, optimal
combining (OC), maximum ratio combining (MRC), as well as SC. Our main contribution is the
derivation of closed-form expressions for the outage probability with various combining schemes in
the interference-limited cooperative ad hoc networks.
\vspace{-0.5em}
\section{System Model}
Consider a large-scale wireless ad hoc network with transmitter-receiver pairs. The locations of
all transmitting nodes are modeled as a homogeneous PPP denoted by \small{$\Phi=\{x_k\}$
}\normalsize with density \small{$\lambda$ }\normalsize on the plane \small{$\Rd$, }\normalsize
where \small{$x_k$ }\normalsize denotes the location of node \small{$k$. }\normalsize Each transmitter has a
uniquely-associated receiver at a distance \small{$d$ }\normalsize away in a random direction and
is not a part of the PPP. All transmitters or receivers are assumed to be identical and equipped
with one antenna. We consider an interference-limited setting~\cite{Andrew:TCOM:2010}, \ie, the
thermal noise is assumed to be negligible.
\vspace{-0.5em}
\subsection{Channel Model}
Signal propagation is subject to both small-scale multipath fading and large-scale path loss. The
instantaneous channel from node $x_1$ to $x_2$ can hence be modeled as
\small{
\begin{align}
g_{12} = h_{12}\ell(x_1-x_2), \label{eqn:channel gain_definition}
\end{align} }\normalsize
where \small{$h_{12}$ }\normalsize captures the small-scale fading and is modeled as an exponential
random variable (RV) with unit variance\footnote{In what follows, we will use the notation
$x\sim\mathcal{E}(\mu)$ to denote $x$ that is exponentially distributed with mean $\mu$. } and
\small{$\ell(x_1-x_2)=\parallel x_1-x_2\parallel^{\alpha}$ }\normalsize characterizes the
large-scale path loss following power law with path loss exponent $\alpha$.
\vspace{-0.5em}
\subsection{Cooperative Transmission Protocol}
We consider a time-slotted Aloha protocol and restrict the number of hops between any
source-destination \small{($S-D$) }\normalsize pair to be two. Similarly
to~\cite{Poor:JWCOM:2011,Ganti:ISIT:2009}, we adopt a two-stage cooperative transmission protocol
as follows:

\subsubsection{Broadcasting Phase}
In the first stage (even time slots), for a given contention probability, \small{$p$, }\normalsize
for transmission, the active nodes (source nodes) from $\Phi$ transmit and all other nodes in
$\Phi$ remain idle. The locations of the source nodes in this stage follow a homogeneous PPP,
denoted as \small{$\Phi^t$, }\normalsize with intensity \small{$p\lambda$. }\normalsize On the
other hand, the idle nodes form another independent homogeneous PPP denoted as \small{$\Phi^r$,
}\normalsize with intensity \small{$(1-p)\lambda$. }\normalsize

For each source node, a selection region \small{$\mathcal{A}$ }\normalsize is defined\footnote{In practice, selecting a suitable relay from a defined region with a small number of relays is desirable. As the implementation complexity increases with the number of relays, a carefully selected region can be used to take into account both protocol complexity and performance gain.}. The idle
nodes in $\Phi^r$ located inside the selection region are required to listen to the transmission
from this source node. If any of these idle nodes successfully decodes the source message, it
becomes a member of the source's potential relay. In other words, a node belonging to $\Phi^r$
located at $x_2$ is said to be a potential relay of the source node at $x_S$, provided that $x_2$
is inside the selection region of the source node and
\small{
\begin{align}
 \frac{ h_{S2}\ell(x_S-x_2)}{I_{\Phi^t}}\geq\beta,  \label{eqn:SIR_definition}
\end{align} }\normalsize
where \small{$I_{\Phi^t} = \sum_{x_k \in \Phi^t\backslash\{x_S\}} h_{k2}\ell(x_k-x_2)$ }\normalsize
is the aggregate interference received at \small{$x_2$ }\normalsize and \small{$\beta$ }\normalsize is the target
(minimum) signal-to-interference ratio (SIR) for data detection\footnote{Note that for $\beta>1$,
each idle node can decode the message from at most one source node, hence can only serve as the
potential relay for at most one source node~\cite{Ganti:ISIT:2009}.}. A popular choice of such a
selection region is a sector defined by a maximum angle \small{$\phi$ }\normalsize and a maximum distance
\small{$d_s$}\normalsize~\cite{Ganti:ISIT:2009,DiLi:Globcom:2010}. A sectorized selection region is considered in this paper. Nevertheless, performance analysis with other different choices of selection regions such as a circular area with radius \small{$d_s$ }\normalsize can be derived in a similar way.
In the ensuing text, the indicator that~\eqref{eqn:SIR_definition} holds is denoted by \small{$\textbf{1}(x_S\rightarrow x_2 | \Phi^t\backslash\{x_S\})$ }\normalsize.
\subsubsection{Relaying Phase}
In the second stage (odd time slots), each destination node informs only one of the potential relays if
any to retransmit the message with repetition code. Similar to the recent work in~\cite{QOS:RelaySelection:TVT2011}, two relay selection methods are considered,
namely, \emph{best relay selection} and \emph{random relay selection}. In particular, the best relay
selection method chooses the potential relay with the best signal strength to the destination
as\vspace{-0.1em}
\small{
\begin{align}
 \mathbb{R}& = \argmax_{x_k \in \Phro} \left\{h_{k D}\ell(x_k-x_D)\right\},
           \label{eqn:mathematicall_description of the best relay}
\end{align} }\normalsize
where \small{$x_D$ }\normalsize presents the destination location and \small{$\Phro$ }\normalsize denotes the set of
potential relays for the \small{$S-D$ }\normalsize pair. On the other hand, the random relay selection method
randomly selects one out of all potential relays with equal probability to forward the source
message. The motivation behind the use of best and random relay selection is to study the trade-off between performance and complexity of relay selection. Best relay selection which exhibits a superior performance compared to random relay selection has a high implementation complexity since it requires high signalling overheads and channel state information (CSI) from all potential relays. On the other hand, at the expense of some performance loss, random relay selection is particularly suitable for low-complexity relay systems.
We denote the set of all transmitting relays, i.e., the relays selected by all source
nodes, as \small{$\Psi$. }\normalsize

Finally, the signals transmitted by the source and the selected relay are combined at the destination node. We
consider three signal combining techniques: namely, OC, MRC, and SC~\cite{Simon:Digitalcom:2004}, which will be described
in detail in the next section. Note that only the direct S-D link can be used when no potential
relay is available.

To analyze the performance of of the considered wireless random network, we focus on a typical transmitter located
at the origin. Similarly to~\cite{Poor:JWCOM:2011}, we will use a polar coordinate system to
facilitate the analysis in which the typical source, \small{$S$, }\normalsize is located at \small{$x_S=(0,0)$,
}\normalsize its destination, \small{$D$, }\normalsize is at \small{$x_D=(d,0)$, }\normalsize and an arbitrary relay
node is at \small{$x=(r,\theta)$. }\normalsize Hence, \small{$\Phro$ }\normalsize now denotes the
potential relay set of the typical source node at the origin.

\vspace{-0.5em}
\section{Outage Probability}
In this section, we derive the outage probability of the described cooperative transmission
protocol. Firstly, we note that the performance of the typical destination node is subject to two
sets of interferers; all concurrently transmitting source nodes in the broadcasting phase and all
concurrently transmitting relays in relaying phase. Specifically, the interference power seen at
the destination in broadcasting and relaying phases can be, respectively, written as
 \vspace{-.1em}
\small{
\begin{subequations}\label{eqn:interference at destination}
\begin{align}
I_{\Phi^t_D} &= \sum_{x_k \in \Phi^t\backslash\{x_S\}} h_{kD}\ell(x_k-x_D),    \label{eqn:interference at destination_broadcasting}\\
I_{\Psi}  &= \sum_{x_m\in\Psi\backslash\{\mathbb{R}\}} h_{mD}\ell(x_{m}-x_D),
\label{eqn:interference at destination_relaying}
\end{align}
\end{subequations} }\normalsize
where \small{$\mathbb{R}$ }\normalsize is the selected relay node for the typical \small{$S-D$
}\normalsize pair. Similarly, the interference power seen by relay \small{$\mathbb{R}$ }\normalsize in the broadcasting
phase can be written as
\vspace{-0.1 em}
\small{
\begin{align}
I_{\Phi^t_{\mathbb{R}}} &= \sum_{x_n \in \Phi^t\backslash\{x_S\}}
h_{n{\mathbb{R}}}\ell(x_n-x_{\mathbb{R}}).
\end{align} }\normalsize
From here on, we denote the channel gain between \small{$S$ }\normalsize and \small{$D$ }\normalsize as \small{$g_{SD}$, }\normalsize the channel gain between \small{$S$ }\normalsize  and \small{$\mathbb{R}$ }\normalsize as \small{$g_{S\mathbb{R}}$, }\normalsize and the channel gain between \small{$\mathbb{R}$ }\normalsize and \small{$D$ }\normalsize  as
\small{$g_{\mathbb{R}D}$. }\normalsize Furthermore, we denote the outage event for link with channel gain \small{$g$ }\normalsize and interference power \small{$I$ }\normalsize by \small{$\mathcal{O}(g/I,\beta)$. }\normalsize Hence
\small{$\mathcal{O}(g_{SD}/I_{\Phi^t_D},\beta)$, }\normalsize is the outage event over the direct
S-D link and the corresponding outage probability is given by~\cite{Baccelli_ALOHA_IT2006}
\vspace{-0.1 em}
\small{
\begin{align}
P_d\!&=\! 1\!-\!\mathbb{E}\left\{\!\textbf{1}(x_S\!\rightarrow\! x_D| \Phi^t\backslash\{x_S\})\!\right\}
\!=\gamma\left(1,\lcb \beta ^{\delta} d^{2}\right),\label{eqn:successprobability_direct hop}
\end{align} }\normalsize
where \small{$\gamma(\cdot,\cdot)$ }\normalsize is the incomplete gamma function~\cite[Eq.
(8.350.1)]{Integral:Series:Ryzhik:1992}, with \small{$\gamma(1,x)=1-\exp(-x)$ }\normalsize and \small{$\lcb =p\lambda c $, }\normalsize with \small{$c =
\delta \pi \Gamma(\delta)\Gamma(1-\delta)$ }\normalsize and \small{$\delta = \frac{2}{\alpha}$.
}\normalsize
\vspace{-1 em}
\subsection{Optimum Combining}
With optimum combining, signals received from the direct and relayed links are weighted to maximize
the SIR at the destination~\cite[Ch. 11]{Simon:Digitalcom:2004}. This scheme requires the instantaneous CSI of all interferers to be known at the receiver, and hence, demands
significant system complexity. As such we consider OC as an important theoretical benchmark to
quantify the performance of the considered network.

Performance of opportunistic relaying with OC receiver has been investigated in~\cite{Bletsas:TWCOM:2010}, where transmitted
signals in broadcasting and relaying phase are impaired by co-channel interference coming from a deterministic interferer.
The outage performance of optimum combining at a multi-antenna receiver with interferers located
according to a PPP was studied in~\cite{Ali:TWC:2010}. Here we extend the result in~\cite{Ali:TWC:2010} to the case of cooperative communications. For simplicity, we ignore the
correlation between the transmitted signal from any source and that from the corresponding relay in
order to obtain analytically tractable result. Following the derivation in~\cite{Ali:TWC:2010}, the resulting SIR at the typical destination after performing OC is given by
\vspace{-.1em}
\small{
\begin{align}
 SIR_{OC} = \frac{g_{SD}}{I_{\Phi^t_D}}+\frac{g_{\mathbb{R}D}}{I_{\Psi}},
    \label{eqn:SIR_Optimal Combining}
\end{align} }\normalsize
which is the sum of received SIRs from the direct link and the relay link. Therefore, the overall
outage event for the cooperative transmission scheme with OC is~\cite{Larsson:COML:2005}
\vspace{0em}
\small{
\begin{align}
\mathcal{O}\left(\frac{g_{SD}}{I_{\Phi^t_D}},\beta\right)\bigcap\left(\mathcal{O}\left(\frac{g_{S\mathbb{R}}}{I_{\Phi^t_{\mathbb{R}}}},\beta\right)
\bigcup\mathcal{O}\left(SIR_{OC},\beta\right)\right).\label{eqn:Outage event of OC receiver}
\end{align} }\normalsize
We now investigate the outage probability of OC receiver for best relay selection and random relay selection.

\subsubsection{Best Relay Selection}
In best relay selection, assuming non-empty \small{$\Phro$, }\normalsize the relay having the best channel gain of the
forward channel is selected, cf.~\eqref{eqn:mathematicall_description of the best relay}. Since computing the cumulative density function (cdf) of the received SIR from relaying phase does not yield a closed-form expression, we obtain a lower bound for the outage
probability in following proposition.
\begin{proposition}\label{Proposition:Outage OC_BS}
The outage probability for cooperative transmission protocol with OC receiver at the destination and best relay selection is given by
\vspace{0em}
\small{
\begin{align}
&P_{out}^{OC,BS} \!=\!\gamma\left(1,\lcb \beta ^{\delta} d^{2}\right)
\left(\gamma\left(1,\frac{\phi(1-p)\lambda}{\beta^\delta\lcb}\gamma\left(1,\lcb \beta^\delta d_s^2\right)\right)\right.\nonumber\\
&\qquad\qquad\qquad\qquad\qquad\left.\times\left(1-P_{OC}^{BS}\right)\!+\!P_{OC}^{BS}\right)\!,\label{eqn:Outage probability of OC receiver_BS}
\end{align} }\normalsize
where \small{${P}_{OC}^{BS}$ }\normalsize denotes the probability that combined SIR drops below the
target SIR and is lower bounded as~\eqref{eqn:P_outage_OC_final_BS} at the top of next page wherein
\small{$d_{\mathbb{R}D} = \sqrt{d^2+r^2-2rd\cos\theta}$ }\normalsize and \small{$\eta_I = \Lambda_I
c$ }\normalsize with \small{$\Lambda_I$ }\normalsize being the intensity of the interferers in the
relaying phase.
\end{proposition}
\small{
\bigformulatop{9}{ \vskip-0.5cm
\begin{align}
P_{OC}^{BS}&\geq
               \gamma(1,\eta \beta^\delta d^2)-\eta\delta d^2 \times\nonumber\\
               &\int_{0}^{\beta}\!y^{\delta-1}\!\exp \left(-\eta y^\delta d^2 \!-\!(1\!-\!p)\lambda\int_{\mathcal{A}}
               \!\exp(-\lcb(\beta\!-\!y)^\delta r^2){\left[1\!-\!\exp(-\!\eta_I(\beta\!-\!y)^\delta d_{\mathbb{R}D}^2)\right]
               rf_{r,\theta}(r,\theta)drd\theta}\!\right)dy,
               \label{eqn:P_outage_OC_final_BS}
\end{align} }\normalsize
\begin{proof}
See Appendix~\ref{proof:Proposition:OC:BS} for the derivation of \small{$P_{OC}^{BS}$ }\normalsize and Appendix~B for
the derivation of \small{$\Lambda_I$. }\normalsize
\end{proof}
Note that~\eqref{eqn:P_outage_OC_final_BS} is a numerically integrable expression
and can be easily evaluated in MAPLE.
\subsubsection{Random Relay Selection}
In random relay selection, assuming non-empty \small{$\Phro$, }\normalsize destination randomly picks a single relay
from \small{$\Phro$. }\normalsize

\begin{proposition}\label{Proposition:Outage OC_RS}
The exact outage probability for cooperative transmission protocol with OC receiver at the destination and random relay selection is given by
\vspace{-.1em}
\small{
\setcounter{equation}{10}
\begin{align}
&P_{out}^{OC,RS} \!=\!\gamma\left(1,\lcb \beta ^{\delta} d^{2}\right)
\left(\gamma\left(1,\frac{\phi(1-p)\lambda}{\beta^\delta\lcb}\gamma\left(1,\lcb \beta^\delta d_s^2\right)\right)\right.\nonumber\\
&\qquad\qquad\qquad\qquad\qquad\times\left.\left(1-P_{OC}^{RS}\right)\!+\!P_{OC}^{RS}\right)\!,\label{eqn:Outage probability of OC receiver_RS}
\end{align} }\normalsize
where \small{${P}_{OC}^{RS}$ }\normalsize is given in~\eqref{eqn:P_outage_OC_final_RS}
at the top of next page.
\end{proposition}
\small{
\bigformulatop{11}{ \vskip-0.5cm
\begin{align}
P_{OC}^{RS}&=
               \gamma(1,\eta \beta^\delta d^2)-\eta\delta d^2
               \int_{0}^{\beta}\int_{\mathcal{A}}\!y^{\delta-1}\!
               \exp \left(-\eta y^\delta d^2 \!- (\beta-y)^\delta\left[\eta_I d_{\mathbb{R}D}^2+\eta r^2\right]\right)rf_{r,\theta}(r,\theta)drd\theta dy,
               \label{eqn:P_outage_OC_final_RS}
\end{align} }\normalsize
\begin{proof}
Following similar steps as in the best relay selection scheme and employing
\vspace{-.1em}
\small{
\setcounter{equation}{12}
\begin{align}
F_{\mathbb{R}D}(\beta)&\!=\!1\!-\!\mathbb{E}\left\{ {\textbf{1}}(x_k\rightarrow x_D \vert \Psi\backslash\{x_\mathbb{R}\})\!\right\}
\label{eqn:success_probability_RS}\\
&\!=\!
        1\!-\!\Prob\left(  h_{kD}\ell(x_k-\ro)>\beta I_\Psi \right\vert x_k \in \Phro)
\nonumber\\
&\!=\!1\!-\!\!
      \int_{\mathcal{A}}\!\!\!\exp\left(\!- \!\beta^\delta\left[\eta_I d_{\mathbb{R}D}^2\!+\!\eta r^2\right]\right)rf_{r,\theta}(r,\theta)drd\theta\!,\nonumber
\end{align} }\normalsize
for the cdf of the received SIR from the relay link, we obtain \small{$P_{out}^{OC,RS}$ }\normalsize which concludes the proof.
\end{proof}
\vspace{0em}
\subsection{Maximal Ratio Combining}
Note that MRC is the optimal combining scheme in the absence of
interference~\cite[Ch. 11]{Simon:Digitalcom:2004}. Although suboptimal in the presence of interference, MRC is an attractive receiver since it does
not require the CSI of the interferers, thus serves as a low complexity alternative compared to OC.

When MRC is employed, the resulting SIR can be expressed as
\vspace{-.5em}
\small{
\begin{align}
 SIR_{MRC} &\!=\! \frac{(g_{SD}+g_{\mathbb{R}D})^2}{g_{SD} I_{\Phi^t_D} + g_{\mathbb{R}D}I_{\Psi}}\!.
    \label{eqn:SIR_Maximal ratio combining}
\end{align} }\normalsize
The overall outage event for the cooperative transmission scheme with MRC can be mathematically written as~\cite{Larsson:COML:2005}
\vspace{-0.7em}
\small{
\begin{align}
\mathcal{O}\left(\frac{g_{SD}}{I_{\Phi^t}},\beta\!\right)\bigcap\left(\mathcal{O}\left(\frac{g_{S\mathbb{R}}}{I_{\Phi^t_{\mathbb{R}}}},\beta\right)
\bigcup\mathcal{O}\left(SIR_{MRC},\beta\right)\right)\!.\label{eqn:Outage probability of MRC receiver}
\end{align} }\normalsize
\subsubsection{Best Relay Selection}
We remark that in this case, with MRC receiver at the destination, derivation of the outage probability is rather involved and thus we only compare the performance using Monte Carlo simulations in Section~\ref{sec:Numerical}.
\subsubsection{Random Relay Selection}
In random relay selection, assuming non-empty \small{$\Phro$, }\normalsize the destination randomly selects one out of all potential relays with equal probability.
\begin{proposition}\label{Proposition:Outage MRC}
The outage probability of the MRC receiver with random relay selection is given by
\vspace{-.1em}
\small{
\begin{align}
&P_{out}^{MRC,RS} \!=\!\gamma\left(1,\lcb \beta ^{\delta} d^{2}\right)
\left(\gamma\left(1,\frac{\phi(1-p)\lambda}{\beta^\delta\lcb}\gamma\left(1,\lcb \beta^\delta d_s^2\right)\right)\right.\nonumber\\
&\qquad\qquad\qquad\qquad\qquad\quad\left.\times\left(1-P_{MRC}^{RS}\right)\!+\!P_{MRC}^{RS}\right)\!,\label{eqn:Outage probability of MRC receiver_RS}
\end{align} }\normalsize
where \small{${P}_{MRC}^{RS}$ }\normalsize is the probability that combined SIR at MRC receiver drops below the target SIR and is given in~\eqref{eqn:P_outage_MRC_final_RS_General} at the top of the next page.

\end{proposition}
\begin{proof}
See Appendix~\ref{proof:Outage MRC}.
\end{proof}

\small{
\bigformulatop{16}{ \vskip-0.5cm
\begin{subequations}\label{eqn:P_outage_MRC_final_RS_General}
\begin{align}
{P}_{MRC}^{RS}  \!&= \! 1 \!- \!\int_{\mathcal{A}}\!\!
                                   \frac{1}{1-\left(\!\frac{d_{\mathbb{R}D}}{d}\!\right)^{\alpha}}
                                    \!\!\left[\!
                                             \exp\left(-d_{\mathbb{R}D}^2 \beta^\delta(\eta+\eta_I+\!\psi)\right)
                                             \!-\!\!  \left(\!\frac{d_{\mathbb{R}D}}{d}\!\right)^{\alpha}
                                                  \!\!\exp\left(\!-d^2 \beta^\delta(\eta\!+\!\eta_I\!+\!\psi)\right)
                                    \right]rf_{r,\theta}(r,\theta)drd\theta,\label{eqn:P_outage_MRC_final_RS}\\
\psi & \!= \!
            \frac{\eta}{\Gamma(\delta)}  \!\left(\frac{d}{d_{\mathbb{R}D}}\right)^{\alpha}
            G_{3 3}^{2 3} \left( \left(\frac{d}{d_{\mathbb{R}D}}\!\right)^{\alpha} \
            \!\Bigg\vert \  \!\!\!{-\delta,-1,0 \atop ~~~0,~~0,-1} \right)+
            \frac{\eta_I}{\Gamma(\delta)}  \!\left(\frac{d_{\mathbb{R}D}}{d}\right)^{\alpha}
            G_{3 3}^{2 3} \left( \!\left(\frac{d_{\mathbb{R}D}}{d}\!\right)^{\alpha} \
            \! \Bigg\vert \  \!\!\!{-\delta,-1,0 \atop ~~~0,~~0,-1} \!\right)
                                  \label{eqn:P_outage_MRC_final_RS_Parameter}
\end{align}
\end{subequations}}\normalsize
\vspace{-0.7em}
\subsection{Selection Combining}
Instead of using OC and MRC which require exact knowledge of the CSI, a system may use SC which
simply requires SIR measurements. Indeed, SC is considered as the least complicated
receiver~\cite[Ch. 11]{Simon:Digitalcom:2004}. With SC receiver at the destination, outage occurs if neither
the direct nor the relayed link can support the target SIR. Hence, the outage event is~\cite{Larsson:COML:2005}
\vspace{-.2em}
\setcounter{equation}{17}
\vspace{-0.2em}
\small{
\begin{align}
\mathcal{O}\left(\frac{g_{SD}}{I_{\Phi^t}},\beta\right)\bigcap\left(\mathcal{O}\left(\frac{g_{S\mathbb{R}}}{I_{\Phi^t_{\mathbb{R}}}},\beta\right)
\bigcup\mathcal{O}\left(\frac{g_{\mathbb{R}D}}{I_{\Psi}},\beta\right)\right)\!,\label{eqn:Outage event of SC receiver}
\end{align} }\normalsize

\subsubsection{Best Relay Selection}
The outage probability is given by
\setcounter{equation}{18}
\vspace{-0.2em}
\small{
\begin{align}
P_{out}^{SC,BS}&=\!\gamma\left(1,\lcb \beta ^{\delta} d^{2}\right)
\left(\gamma\left(1,\frac{\phi(1-p)\lambda}{\beta^\delta\lcb}\gamma\left(1,\lcb \beta^\delta d_s^2\right)\right)\right.\nonumber\\
&\qquad\qquad\qquad\quad\left.\times\left(1-P_{SC}^{BS}\right)\!+\!P_{SC}^{BS}\right),\label{eqn:Outage probability of SC RS}
\end{align} }\normalsize
\begin{figure}[h]
\centering
\includegraphics[width=85mm, height=60mm]{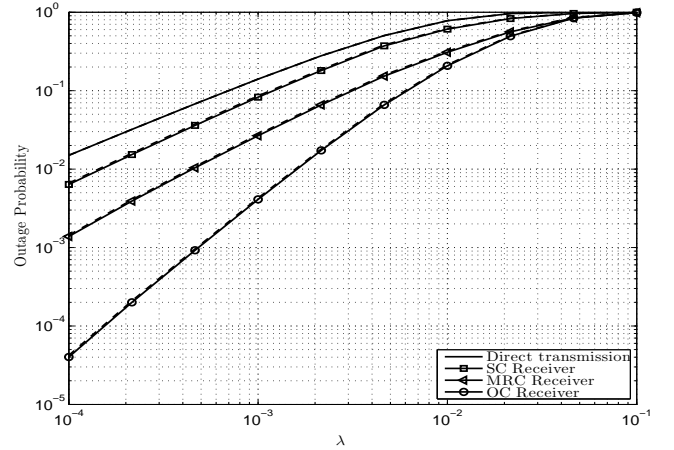}
\vspace{-0.8 em}
\caption{Outage probability versus $\lambda$ for different combining receivers with random relay selection. Simulation results are shown with dashed lines.}
\label{fig:outage_OC_MRC_OC_RandomSelection}
\vspace{-1.0 em}
\end{figure}
where
\small{
\begin{align}
&{P}_{SC}^{BS}\!= 1\!-\mathbb{E}\left\{
                             {\textbf{1}}(x_k\rightarrow x_D \vert \Psi\backslash\{x_\mathbb{R}\})\right\}.
                             \label{eqn:outage_probability SC}
\end{align} }\normalsize
Note that we obtained a lower bound for \small{${P}_{SC}^{BS}$ }\normalsize
in~\eqref{eqn:CDF_SIR_RD} which
yields a lower bound for the outage probability of the SC receiver at the destination with best
relay selection.
\subsubsection{Random Relay Selection}
The outage probability of random relay selection follows from~\eqref{eqn:Outage probability of SC RS} by replacing \small{$P_{SC}^{BS}$ }\normalsize with \small{$P_{SC}^{RS}$ }\normalsize, given in~\eqref{eqn:success_probability_RS}.
\section{Numerical and Simulation Results}\label{sec:Numerical}
In this section, we study the accuracy of the derived analytical results and compare the outage probability of OC, MRC and SC schemes with best and random relay selection. In all simulations, we have set \small{$\alpha=4$, $\beta=3$ }\normalsize dB, \small{$d=10$ }\normalsize m, \small{$d_s=7$  }\normalsize m
and \small{$\phi=\frac{\pi}{3}$, }\normalsize unless stated otherwise. To ensure a fair comparison,
the SIR threshold of cooperative transmission is set to be twice as much as direct transmission
threshold. This is because, in the broadcasting phase, the source uses half of the channel uses and in the relaying phase, the relay uses the remaining channel uses.

\begin{figure}[h]
\centering
\includegraphics[width=85mm, height=60mm]{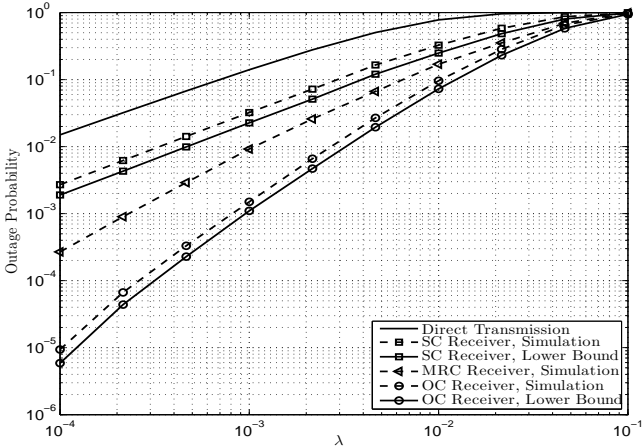}
\vspace{-0.8 em}
\caption{Outage probability versus $\lambda$ for different combining receivers with best relay selection .}
\label{fig:outage_OC_MRC_OC_BestSelection}
\vspace{-0.9 em}
\end{figure}
The performance of OC and sub-optimal combining schemes (MRC and SC) can be
further ascertained by referring to Fig.~\ref{fig:outage_OC_MRC_OC_RandomSelection}, where the
probability of outage as a function of node density, \small{$\lambda$ }\normalsize is shown for random relay selection. Analytical expressions in~\eqref{eqn:P_outage_OC_final_RS},~\eqref{eqn:P_outage_MRC_final_RS_General}, and
~\eqref{eqn:outage_probability SC} are confirmed as they are seen to follow the simulations tightly. The performance improvement of all three combining schemes compared to the direct transmission is noticeable for all intensities. Note that although increasing the intensity of nodes results in increasing the number of qualified relay nodes, the interference caused by these relay nodes in relaying phase is increased. Hence, for higher values of intensity the performance of all schemes are converging to the same values.

The effect of best relay selection on the outage performance of combining schemes
is investigated in Fig.~\ref{fig:outage_OC_MRC_OC_BestSelection} and the tightness of proposed lower bound for OC and SC receiver are validated. Comparing Fig.~\ref{fig:outage_OC_MRC_OC_RandomSelection} and Fig~\ref{fig:outage_OC_MRC_OC_BestSelection} reveals that best relay selection significantly outperforms the random relay selection as expected.

Our observation of the relation between the outage performance of combining schemes and selection region parameters which due to the space limitation are not shown in simulations, reveals that:
\emph{1)} There is an optimal value of the parameter \small{$\phi$ }\normalsize for each combining scheme that achieves the maximum success probability.
\emph{2)} Increasing the parameter \small{$\phi$ }\normalsize beyond its optimum value does not degrade the success probability. The reason is that although enlarging the selection region increases the possibility of being relay nodes in this area and consequently increases the intensity of interferer set for the second hop, the relay selection strategy only selects a single relay, and thus the intensity of the interferers for the second hop is upper bounded. Therefore, the success probability remains constant.

\section{Conclusion}
The performance of relay selection schemes along with different combining schemes for cooperative
transmissions in ad hoc networks have been studied. In particular, we have obtained closed-form
expressions for the outage probability of OC, MRC, and SC receivers at the destination with random relay selection. We have also derived two useful tight lower bounds for the OC and SC receivers with best relay selection. The accuracy of the analytical
results has been validated using Monte Carlo simulations.
\appendices
\vspace{-0.5em}
\section{Proof of Proposition~\ref{Proposition:Outage OC_BS}}
\label{proof:Proposition:OC:BS}
The proof of Proposition~\ref{Proposition:Outage OC_RS} is accomplished in four steps as
follows:\\
\textbf{Step 1:}
The outage probability corresponding to the event \small{$\mathcal{O}\left(g_{SD}/I_{\Phi^t_D},\beta\right)$ }\normalsize is given
by~\eqref{eqn:successprobability_direct hop}.
\\
\textbf{Step 2:} The outage event \small{$\mathcal{O}\left(g_{S\mathbb{R}}/I_{\Phi^t_{\mathbb{R}}},\beta\right)$ }\normalsize is equivalent to the event that
the potential relay set \small{$\Phro$ }\normalsize  is empty. Note that \small{$\Phro$ }\normalsize which is derived by thinning the PPP \small{$\Phi^{r}$, }\normalsize is still a PPP and
Marking theorem of Poisson processes~\cite{StochasticGeometry_Book_1996} gives its density as
\small{
\begin{align}\nonumber
   \Lambda_o(x)  & = (1-p)\lambda \mathbb{E}\left\{ \textbf{1}(x_S\rightarrow x_k | \Phi^t\backslash\{x_S\})\right\}
   \nonumber\\
                 & =(1-p)\lambda\exp\left(-\lcb\beta^\delta \ r^2 \right).\label{eqn:intensity Potential_relay set}
\end{align} }\normalsize
Therefore, the outage probability corresponding to desired outage event is given by
\small{
\begin{align}
    P_r   = \exp(-\mu_o(\mathcal{A})),
               \label{eqn:outage probability SR event}
\end{align} }\normalsize
where \small{$\mu_o(\mathcal{A})$ }\normalsize denotes the mean measure of \small{$\Phro$ }\normalsize and
\small{
\begin{align}
    \mu_o(\mathcal{A})   = \int_{\mathcal{\mathcal{A}}}\Lambda_o(dx),
             \label{eqn:mean measure of potential relay set}
\end{align} }\normalsize
in which \small{$\mathcal{A}$ }\normalsize is the selection region and $dx$ denotes a two-dimensional variable of integration over the polar area.
Finally, by substituting~\eqref{eqn:intensity Potential_relay set} into~\eqref{eqn:mean measure of potential relay set}
the outage probability in~\eqref{eqn:outage probability SR event} is obtained for sectorized selection region as
\small{
\begin{align}
        P_r   =   \exp \left(-\frac{\phi(1-p)\lambda}{\beta^\delta\lcb}\gamma\left(1,\lcb \beta^\delta d_s^2\right)\right).
        \label{eqn:outage event empty potential relay set}
\end{align} }\normalsize
\\
\textbf{Step 3:}
The outage probability corresponding to the outage event \small{$\mathcal{O}\left(SIR_{OC},\beta\right)$ }\normalsize can be written as\footnote{In general, correlation of node locations in wireless ad hoc network, makes the interference temporally and spatially correlated~\cite{Heng:Gong12tmc}.
However, our derivation in~\eqref{eqn:Pout_Optimal_combining} does not include the impact of correlation. The justification of this assumption is primarily to preserve analytical tractability and simplicity, however, simulations validate this assumption.}
\vspace{-0.2em}
\small{
\begin{align}
 {P}_{OC}^{BS}\! &=\Prob\left(SIR_{SD}+SIR_{\mathbb{R}D}<\beta\right)\nonumber\\
 &=\! \int_0^{\beta}\!F_{\mathbb{R}D}(\beta\!-\!y)f_{{SD}}(y)dy,
    \label{eqn:Pout_Optimal_combining}
\end{align} }\normalsize
where \small{$F(\cdot)$ }\normalsize and \small{$f(\cdot)$ }\normalsize denotes cdf and probability density function (pdf) of the RV,
respectively. \small{$f_{SD}(y)$ }\normalsize can be found simply by taking the first order derivation of~\eqref{eqn:successprobability_direct hop}, which yields
\vspace{-0.2em}
\small{
\begin{align}
    f_{SD}(y) = \eta d^{2}\delta y^{\delta-1}\exp(-\eta d^{2}y^{\delta}).\label{eqn:PDF_SIR_Sd}
\end{align} }\normalsize
A lower bound for \small{$F_{\mathbb{R}D}(\beta)$ }\normalsize is obtained as
\small{
\begin{align}
&F_{\mathbb{R}D}(\beta) =1-\mathbb{E}\left\{ {\textbf{1}}(x_k\rightarrow x_D \vert \Psi\backslash\{x_\mathbb{R}\})\!\right\}\nonumber\\
&~
          \!=\!1\!-\!\Prob\left( \max_{x_k \in \Phro} \left\{ h_{kD}\ell(x_k-x_D)\right\}>\beta I_\Psi\right)\nonumber\\
&~
          \!=\!1\!-\!\mathbb{E}\left\{ \prod_{x_k \in \Phro}\exp\left(-\frac{\beta I_\Psi}{\ell(x_k-x_D)}\right)\right\}\nonumber\\
&~
          \!\stackrel{(a)}{=}\!
          1\!-\!\mathbb{E}\left\{
                               \exp\! \left(\!-\!
                               \int_{\mathcal{A}}\!
                               {\!\left[1\!-\!\exp\!\left(\!-\!\frac{\beta I_\Psi}{\ell(x_k\!-\!x_D)}\right)\!\right]
          \Lambda_o(dx)}\!\right)\!\right\}\nonumber\\
&~
          \!\stackrel{(b)}{\geq}\!
               1\!-\!\exp
               \left(-\int_{\mathcal{A}}\!
                                      {\!\left[1\!-\!\mathbb{E}\left\{\exp\left(-\frac{\beta I_\Psi}{\ell(x_k\!-\!\ro)}\right)\right\}\right]
                                      \Lambda_o(dx)}\right)
                \nonumber\\
&~
           \!\stackrel{(c)}{=}\!
               1\!-\!\exp \left(-\int_{\mathcal{A}}\!
                                   {\left[1-\exp(-\eta_I \beta^\delta d_{\mathbb{R}D}^2)\right]\Lambda_o(dx)}\right),
\label{eqn:CDF_SIR_RD}
\end{align} }\normalsize
where (a) follows from the generating functional of the PPP, \small{$\Phro$ }\normalsize with intensity
\small{$\Lambda_o(x)$}\normalsize\footnote{Let $\nu(x):\Rd\rightarrow[0,1]$ and $\int_{\Rd}{\vert1-\nu(x)\vert dx}<\infty$.
When $\Phi$ is Poisson of intensity $\lambda$, the conditional generating functional is $\mathbb{E}\{\prod_{x\in\Phi}\nu(x)\}=\exp\left(-\lambda\int_{\Rd}[1-\nu(x)]dx\right)$.}
and (b) follows by using the the Jensen's inequality. (c) holds by taking the expectation over \small{$I_\Psi$ }\normalsize where the exponential distribution of channel gains and the generating functional of \small{$\Psi$ }\normalsize have been used.
Moreover, \small{$\eta_I = c\Lambda_I$ }\normalsize where \small{$\Lambda_I$ }\normalsize is the intensity of interferers for second hop transmission.
The details of intensity evaluation for the second hops' interferers are deferred to Appendix~\ref{Intensity of Interfere Sets for Second Hop Transmission}.
Plugging~\eqref{eqn:CDF_SIR_RD} together with~\eqref{eqn:PDF_SIR_Sd} into~\eqref{eqn:Pout_Optimal_combining}, yields the lower bound on \small{$P_{OC}^{BS}$ }\normalsize in~\eqref{eqn:Pout_Optimal_combining}.\\
\\
\textbf{Step 4:} Referring to the outage event in~\eqref{eqn:Outage event of OC receiver},
the overall outage probability of the cooperative transmission with OC receiver and best relay selection
is given by
 \vspace{-0.2em}
\small{
\begin{align}
  P_{out}^{OC,BS}=P_d\left(\left(1-P_r\right)\left(1-P_{OC}^{BS}\right)+P_{OC}^{BS}\right)
               \label{eqn:outage_OC_BS_Raw}
\end{align} }\normalsize
Plugging~\eqref{eqn:successprobability_direct hop},~\eqref{eqn:outage event empty potential relay set} and~\eqref{eqn:Pout_Optimal_combining}
into~\eqref{eqn:outage_OC_BS_Raw} and after some manipulations gives the desired result in~\eqref{eqn:Outage probability of OC receiver_RS}.
\vspace{-0.8em}
\section{Intensity of the Interferer set}
\label{Intensity of Interfere Sets for Second Hop Transmission}
Deriving the intensity of interferer set for the second-hop transmission is equivalent to exploit how many selected relays are there for transmission in the second hop.
Provided that the potential relay set is not empty, since each source only selects a single relay in both best relay selection and random relay selection, it is clear that the intensity of interferer set is at most $p\lambda$, which is the intensity of the source nodes.
Let us denote the probability that an arbitrary source node successfully selects one relay node by $q$.
Then, the intensity of active relay nodes in second stage of cooperative transmission protocol is $pq\lambda$, which is also the intensity of the interferers.
Note that $q$ is proportional to the probability that at least there is one relay in potential relay set which for a typical pair is mathematically
described by
\vspace{-0.2em}
\small{
\begin{align}
    q   = 1-P_r = 1-\exp(-\mu_o(\mathcal{A})).
               \label{eqn:one relay in potential relay set}
\end{align} }\normalsize
Therefore, in the case of a sectorized selection region, the intensity of interferer set in second stage of transmission, $\Lambda_I$, is given by
 \vspace{-0.2em}
\small{
\begin{align}
\Lambda_I &   = p \lambda\left(1- \exp\left(-(1-p)\lambda\int_{\mathcal{A}}{\exp\left(-\lcb\beta^\delta \ r^2\right)dx }\right)\right)
\nonumber\\
&=p \lambda\gamma\left(1,\frac{\phi(1-p)\lambda}{\lcb \beta^\delta}\gamma\left(1,\lcb \beta^\delta d_s^2\right)\right).
    \label{eqn:Intensity of interferer set for the second-hop transmission}
\end{align} }\normalsize
\scriptsize{
\bigformulatop{35}{ \vskip-0.5cm
\begin{align}
{P}_{MRC}^{RS} &=
               1 -\mathbb{E}\left\{
                                    \frac{1}{\mu_1-\mu_2}
                                     \left[\mu_1
                                      \exp\left(-\mu_2\beta\left(
                                                                 \frac{u}{u+v}I_{\Phi^t_D}+\frac{v}{u+v}I_{\Psi}\right)\right) - \right.\right.\nonumber\\
                                     & \qquad\qquad\qquad\qquad\qquad
                                     \left.\left.
                                      \mu_2
                                      \exp\left(-\mu_1\beta\left(
                                                                 \frac{u}{u+v}I_{\Phi^t_D}+\frac{v}{u+v}I_{\Psi}\right)
                                                                 \right)\right]\right\}\nonumber\\
                                  &\stackrel{(a)}{=}
                                  \! 1 \!- \!\mathbb{E}\left\{
                                   \frac{1}{\mu_1\!-\!\mu_2}
                                  \!\left[ \mu_1
                                   \!\!\prod_{x_k\in\Phi^t\backslash\{x_S\}}\!\!\!\!
                                                        \mathcal{M}_{W_1}\!\left(\mu_2\beta\ell(x_k-x_D)\right )
                                    \! \!\prod_{y_k\in \Psi\backslash\{\mathbb{R}\}}\!\!\!\!
                                                        \mathcal{M}_{W_2}\!\!\left(\mu_2\beta\ell(y_k-x_D)\right)
                                                        \!- \nonumber\right.\right.\\
                                                        &\qquad\qquad\qquad\qquad\left.\left.\mu_2
                                    \! \!\prod_{x_k\in\Phi^t\backslash\{x_S\}}\!\!\!\!
                                                        \mathcal{M}_{W_1}\!\!  \left(\mu_1\beta\ell(x_k-x_D)\right)
                                    \! \!\prod_{y_k\in \Psi\backslash\{\mathbb{R}\}}\!\!\!
                                                       \mathcal{M}_{W_2}\! \! \left(\mu_1\beta\ell(y_k-x_D)\right)\! \right]\! \right\}
                                    \label{eqn:P_outage_MRC_proof}
\end{align} }\normalsize
\vspace{-1.2em}
\section{Proof of Proposition~\ref{Proposition:Outage MRC}}
\label{proof:Outage MRC}
Before deriving the outage probability we will need the following lemma.
\begin{Lemma}\label{lem: MGF of the W}
Given three RVs \small{$u\sim \mathcal{E}(\mu_1)$, $v\sim \mathcal{E}(\mu_2)$, }\normalsize and
\small{$z\sim \mathcal{E}(1)$, }\normalsize the moment generating function (MGF) of \small{$W_1 =
uz/(u+v)$, $\mathcal{M}_{W_1}(s)=\mathbb{E}\{\exp(-sW_1)\}$ }\normalsize is given by
 \vspace{-0.2em}
\small{
\begin{align}
\mathcal{M}_{W_1}(s) & \!=\! 1-\!\frac{s}{s+1}\times\nonumber\\
&\quad\quad\left[\!1\!-\!\frac{\mu_1}{\mu_2}\frac{1}{s\!+1}G_{2 2}^{2 2} \left( \frac{\mu_1}{\mu_2}\frac{1}{s+\!1} \  \Bigg\vert \  {\!\!\!-1,0 \atop 0,0} \right)\right]
,    \label{eqn:MGF_w1_MeijerG}
\end{align} }\normalsize
where \small{$G_{p q}^{m n} \left( s \  \vert \  {a_1\cdots a_p \atop b_1\cdots b_q} \right)$ }\normalsize denotes the Meijer's G-function defined in~\cite[ Eq. (9.301)]{Integral:Series:Ryzhik:1992}. With appropriate changes of the mean indices in~\eqref{eqn:MGF_w1_MeijerG}, the MGF of \small{$W_2 = vz/(u+v)$, }\normalsize can be expressed as
 \vspace{-0.2em}
\small{
\begin{align}
\mathcal{M}_{W_2}(s) & \!=\! 1-\!\frac{s}{s+1}\times\nonumber\\
&\quad\quad\left[\!1\!-\!\frac{\mu_2}{\mu_1}\frac{1}{s\!+1}G_{2 2}^{2 2} \left( \frac{\mu_2}{\mu_1}\frac{1}{s+\!1} \  \Bigg\vert \  {\!\!\!-1,0 \atop 0,0} \right)\right]
,    \label{eqn:MGF_w2_MeijerG}
\end{align} }\normalsize
\end{Lemma}
\begin{proof}
See Appendix~\ref{proof:MGF of Wi}
\end{proof}
Following similar steps as in the OC case with random relay selection, we get the outage probability as~\eqref{eqn:Outage probability of MRC receiver_RS}.

What remains to calculate is then to determine the outage probability \small{${P}_{MRC}^{RS}$}\normalsize. Let us define
\small{
\begin{align}
u \triangleq g_{SD}, \quad
v \triangleq g_{\mathbb{R}D}.
    \label{eqn:Auxiliary_Variable}
\end{align} }\normalsize
The RVs, \small{$u$ }\normalsize and \small{$v$ }\normalsize are exponentially distributed with parameter \small{$\mu_1=d^{\alpha}$ }\normalsize and \small{$\mu_2=\drd^{\alpha}$,} \normalsize respectively.
We recall that for random relay selection case, the RV \small{$g_{\mathbb{R}D}$ }\normalsize is considered as an arbitrary exponential RV. It can be shown that the SIR in~\eqref{eqn:SIR_Maximal ratio combining} is of the form of~\cite[Eq. (28)]{Haimovich:TVT:2000}, i.e.,
\small{
\begin{align}
SIR_{MRC} &= \frac{\vert \vc_s\vert^2}{\sum_{i=1}^{L} \vert \boldsymbol {\nu}_i\vert^2},
                                    \label{eqn:SIR_MRC_hamovich}
\end{align} }\normalsize
where \small{$\vert \vc_s\vert^2 = u+v$,~$\boldsymbol {\nu}_i =\vc_s^\dag\vc_i/\vert \vc_s\vert$}\normalsize, and \small{$L=\vert\Phi^t\vert+\vert\Psi\vert$ }\normalsize with \small{$ \vc_i$ }\normalsize being the channel coefficient between the interferer \small{$i$ }\normalsize and destination and \small{$\vert\cdot\vert$ }\normalsize being the cardinality of a set. Therefore, it can be shown that \small{$\boldsymbol{\nu}_i$}\normalsize's are independent of \small{$\vc_s$}\normalsize~\cite{Haimovich:TVT:2000} and thus \small{${P}_{MRC}^{RS}$ }\normalsize can be written as
\small{
\begin{align}
{P}_{MRC}^{RS} &= \Prob\left(
                      (u+v)<\beta\left[\frac{u}{u+v}I_{\Phi^t_D} + \frac{v}{u+v}I_{\Psi} \right]\right).
                                    \label{eqn:P_outage_MRC_hamovich}
\end{align} }\normalsize
Using the cdf of \small{$u+v$ }\normalsize given in~\cite[Eq. (40)]{Laneman:IT:Dec:2004}, the outage probability can be expressed as~\eqref{eqn:P_outage_MRC_proof} at the top of this page, where ($a$) follows by taking the expectation with respect to \small{$W_1$ }\normalsize and \small{$W_2$. }\normalsize Using Lemma~\ref{lem: MGF of the W}, the generating functional of Poisson processes, $\Phi^t$ and $\Psi$, and~\cite[ Eq. (3.194.4) and Eq. (7.811.2)]{Integral:Series:Ryzhik:1992} gives, after some manipulation, the desired result in~\eqref{eqn:P_outage_MRC_final_RS_General}.
\vspace{-0.5em}
\section{Proof of Lemma~\ref{lem: MGF of the W}}
\label{proof:MGF of Wi}
In order to obtain the MGF of \small{$W_1$}\normalsize, we first derive the cdf of \small{$W_1$ }\normalsize as
\vspace{-0.25em}
\small{
\setcounter{equation}{36}
\begin{align}
F_{W_1}(w)  &\!=\! \int_{0}^{\infty} \Prob\left(z<w\left(1+\gamma\right)\right)f_\Upsilon(\gamma)d\gamma,\label{eqn:CDF_w1}
\end{align} }\normalsize
where \small{$\gamma\triangleq\frac{u}{v}$, }\normalsize and its pdf can be readily obtained as
\vspace{-0.25em}
\small{
\begin{align}
f_{\Upsilon}(\gamma) = \frac{\frac{\mu_1}{\mu_2}}{(\gamma+\frac{\mu_1}{\mu_2})^2}.    \label{eqn:PDF_V}
\end{align} }\normalsize
Plugging~\eqref{eqn:PDF_V} into~\eqref{eqn:CDF_w1}, using~\cite[ Eq. (5.1.4) and Eq. (13.6.30)]{Abramowitz_Handbook_1970}, and after some manipulations, we get the cdf in closed-form as
\vspace{-0.25em}
\small{
\begin{align}
F_{W_1}(w)  &\!=1\!-\!\exp(-w)\left[1-\frac{\mu_1}{\mu_2}w\Psi\left(1,1,\frac{\mu_1}{\mu_2}w\right)\right],
              \label{eqn:CDF_w1_Tricomi}
\end{align} }\normalsize
where \small{$\Psi(\cdot,\cdot,\cdot)$ }\normalsize denotes the Tricomi confluent hypergeometric
function~\cite[ Eq. (9.211.4)]{Integral:Series:Ryzhik:1992}. Now, the MGF of \small{$W_1$}\normalsize~can be directly found from
\vspace{-0.25em}
\small{
\begin{align}\nonumber
\mathcal{M}_{W_1}(s) & = s \mathcal{L}(F_{W_1}(w))-F_{W_1}(0),
\end{align} }\normalsize
where \small{$\mathcal{L}(\cdot)$ }\normalsize denotes the Laplace transform and \small{$F_{W_1}(0)=0$. }\normalsize
Using~\cite[ Eq. (3.36.1.7)]{Laplace:Prudnikov:1992} one can obtain the MGF of \small{$W_1$}\normalsize~in closed-form as~\eqref{eqn:MGF_w1_MeijerG} and the lemma is proved.
\vspace{-0.25em}
\section*{Acknowledgment}
This work was supported in part by the Australian Research Council's Discovery Projects funding scheme (project no. DP110102548).
\bibliographystyle{IEEEtran}
\bibliography{IEEEabrv,refrence_AdHoc_AHS}

\end{document}